\documentclass[10pt, conference, letterpaper]{IEEEtran}
\usepackage{filecontents}

\usepackage{amsthm}
\newtheorem{theorem}{Theorem}

\newtheorem{definition}{Definition}

\newtheorem*{Game*}{Game}
\newtheorem*{TokGen*}{TokenGeneration Phase}
\newtheorem*{TokVer*}{TokenVerification Phase}
\newtheorem*{Setup*}{Setup}
\newtheorem*{Query*}{Query}
\newtheorem*{Challenge*}{Challenge}
\newtheorem*{Guess*}{Guess}
\usepackage{breqn}
\theoremstyle{remark}

\usepackage{algorithm}
\usepackage{algpseudocode}
\usepackage{comment}
\usepackage{hyperref}
\hypersetup{
    hidelinks = true
}

\usepackage{tikz}
\usetikzlibrary{quantikz}

\usepackage[capitalize]{cleveref}
\usepackage{tensor}
\usepackage{algpseudocode}
\usepackage{mathtools,physics, bbm}
\usepackage{graphicx}
\usepackage{dcolumn}
\usepackage{bm}

\begin{document}
\bstctlcite{IEEEexample:BSTcontrol}

\title{Optimal depth and a novel approach to variational quantum process tomography}

\author{
\IEEEauthorblockN{Vladlen Galetsky$^{*}$, Pol Julià Farré\IEEEauthorrefmark{2}, Soham Ghosh$^{*}$, Christian Deppe\IEEEauthorrefmark{2} and Roberto Ferrara$^{*}$} 
\IEEEauthorblockA{$^{*}$Technical University of Munich\\
\IEEEauthorrefmark{2}Technical University of Braunschweig \\
Email: vladlen.galetsky@tum.de, pol.julia-farre@tu-braunschweig.de, soham.ghosh@tum.de,  \\ christian.deppe@tu-braunschweig.de, roberto.ferrara@tum.de}
}

\maketitle
\begin{abstract}
In this work, we present two new methods for Variational Quantum Circuit (VQC) Process Tomography onto $n$ qubits systems: PT\_VQC and U-VQSVD. 

Compared to the state of the art, PT\_VQC halves in each run the required amount of qubits for process tomography and decreases the required state initializations from \(4^{n}\) to just $2^{n}$, all while ensuring high-fidelity reconstruction of the targeted unitary channel \(U\). It is worth noting that, for a fixed reconstruction accuracy, PT\_VQC achieves faster convergence per iteration step compared to Quantum Deep Neural Network (QDNN) and tensor network schemes.

The novel U-VQSVD algorithm utilizes variational singular value decomposition to extract eigenvectors (up to a global phase) and their associated eigenvalues from an unknown unitary representing a general channel. We assess the performance of U-VQSVD by executing an attack on a non-unitary channel Quantum Physical Unclonable Function (QPUF). U-VQSVD outperforms an uninformed impersonation attack (using randomly generated input states) by a factor of 2 to 5, depending on the qubit dimension.

For the two presented methods, we propose a new approach to calculate the complexity of the displayed VQC, based on what we denote as optimal depth. 
\end{abstract}

\maketitle

\section{Introduction}
\subsection{Process tomography}
Quantum tomography plays an essential role in the security assessment and characterization of channel and state evolution in Noisy Intermediate Scale Quantum (NISQ) devices \cite{Chen2023, RevModPhys.94.015004}. A class of promising candidates, ranging from Quantum Deep Neural Networks (QDNN) \cite{Beer2020}, tensor networks \cite{Torlai2023}, Variational Quantum Circuit (VQC) \cite{PhysRevA.105.032427,PhysRevResearch.5.L032040}
and quantum compilation algorithms \cite{Hai2023}, try to perform state and process tomography beyond the current classical computation capabilities.

In process tomography, certain approaches presume particular symmetries within the target circuit. For instance, tensor networks \cite{Torlai2023} rely on the notion of low entanglement structures. In unsupervised learning and QDNN, algorithms extract multidimensional probability distributions from raw data. The considerable expressivity of such models enables the retention of assumptions regarding high entanglement degrees in the unitary under study \cite{PhysRevResearch.3.L032049}.

As the system to be queried scales, state and unitary reconstruction become resource-intensive tasks, often resulting in an exponential overhead \cite{Moll_2018}. The poor scalability, the lengthy circuit depths and the extensive number of initializations and iterations required to minimize the objective function render VQC schemes unreliable for NISQ device applications. To mitigate these challenges, we propose in Section.\,\ref{Process tomography2} a new VQC process tomography scheme called PT\_VQC.


\subsection{Process tomography of unitary evolutions}

Studying non-unitary channels presents a more intricate task. One proposed solution introduces convex optimization \cite{HUANG2020286} to address this challenge. However, this approach relies on understanding the constraints of the process matrix and assumes the structure of the channel under examination. Another proposed method utilizes classical shadows \cite{PhysRevResearch.6.013029}, which are based on shadow tomography techniques. However, the classical shadow post-processing becomes highly inefficient as the amount of qubits \( n \) scales up. Additionally, for Clifford shadows, prior knowledge of the subsystem is required to perform the non-unitary reconstruction.

We instead introduce a novel solution in Section\,\ref{Unitary singular value decomposition}: the U-VQSVD scheme. We showcase its effectiveness by executing an impersonation attack against a non-unitary Quantum Physical Unclonable Function (QPUF) channel, reconstructing the singular values and vectors of the unknown unitary evolution defining it.

\subsection{Quantum Physical Unclonable Function (QPUF)}
QPUFs are hardware-based unclonable devices with an inherit randomness \cite{Mina21,cry,Pirnay_2022} that can be produced, for example, by having an imperfection within an optical crystal. Based on a set of challenges, such randomness results into a production of a unique and unpredictable set of responses and this randomized mapping can be harnessed within the realm of secure authentication.

     Recent research \cite{Mina21} introduced a mathematical framework that, for the first time, provided a QPUF scheme with provable security against quantum adversaries, referred to as "Quantum selective unforgeability". However, according to \cite{Mina21} and \cite{GGDF22}, the security notion they define under the label "Quantum existential unforgeability" is still nonexistent in a unitary QPUF. Subsequently, \cite{QPUF_part} established an analogous notion of existential unforgeability for their designed non-unitary channel-based QPUF.
     
     We numerically show that U-VQSVD does not break the security notions of \cite{QPUF_part}, while still performing better than an attacker employing random states during authentication.

\section{Our Contributions}

We divide our set of contributions into two main parts. The former arises from a new process tomography algorithm called PT\_VQC, found in Section.\,\ref{Process tomography2}. In contrast, the approach proposed in \cite{PhysRevA.105.032427} relies on an inefficient SWAP-test algorithm to define their cost function, necessitating $4^{n}$ state initializations for a complete reconstruction of the target unitary. Such high requirements for VQC process tomography schemes render them slow and unreliable for NISQ device applications, particularly those reliant on qubit dimensionality.

We instead deliver an enhanced scheme by introducing an alternative quantum circuit architecture and redefining the cost function. Notably, we reduce the required state initializations from $4^{n}$ to just $2^{n}$ and halve for each run the required number of qubits. As a consequence, we improve the processing time of the algorithm while, in parallel, enhancing the accuracy of its gradient computation by applying the 4-term shift rule. We conducted a comparative analysis of the efficiency and cost of the PT\_VQC algorithm against variational Matrix Product Operators (MPO) based on tensor networks \cite{Torlai2023}, as well as against QDNN \cite{Beer2020}.

For the second main contribution, we introduce a new algorithm called U-VQSVD, found in Section.\,\ref{Unitary singular value decomposition}.
U-VQSVD is inspired by the variational quantum singular value decomposition scheme proposed in \cite{Wang_2021}.
 Our novel proposal allows us to learn unknown general channels by efficiently reconstructing their eigenvalues and eigenvectors up to a phase. Furthermore, we evaluate the effectiveness of our new process tomography scheme by constructing an impersonation attack on the authentication protocol of a Phase Estimation-based QPUF (PE-QPUF)  \cite{QPUF_part}. This assessment involves comparing the capabilities of an attacker employing random states to those of an attacker utilizing U-VQSVD to generate input states, alongside a trusted party.
 
 Finally, we discuss how to select the depth of the VQC introducing the concept of optimal depth and providing the reader with a guideline for its estimation.



\begin{figure}
\hspace*{0.85cm}
\includegraphics[width=0.35\textwidth,clip]{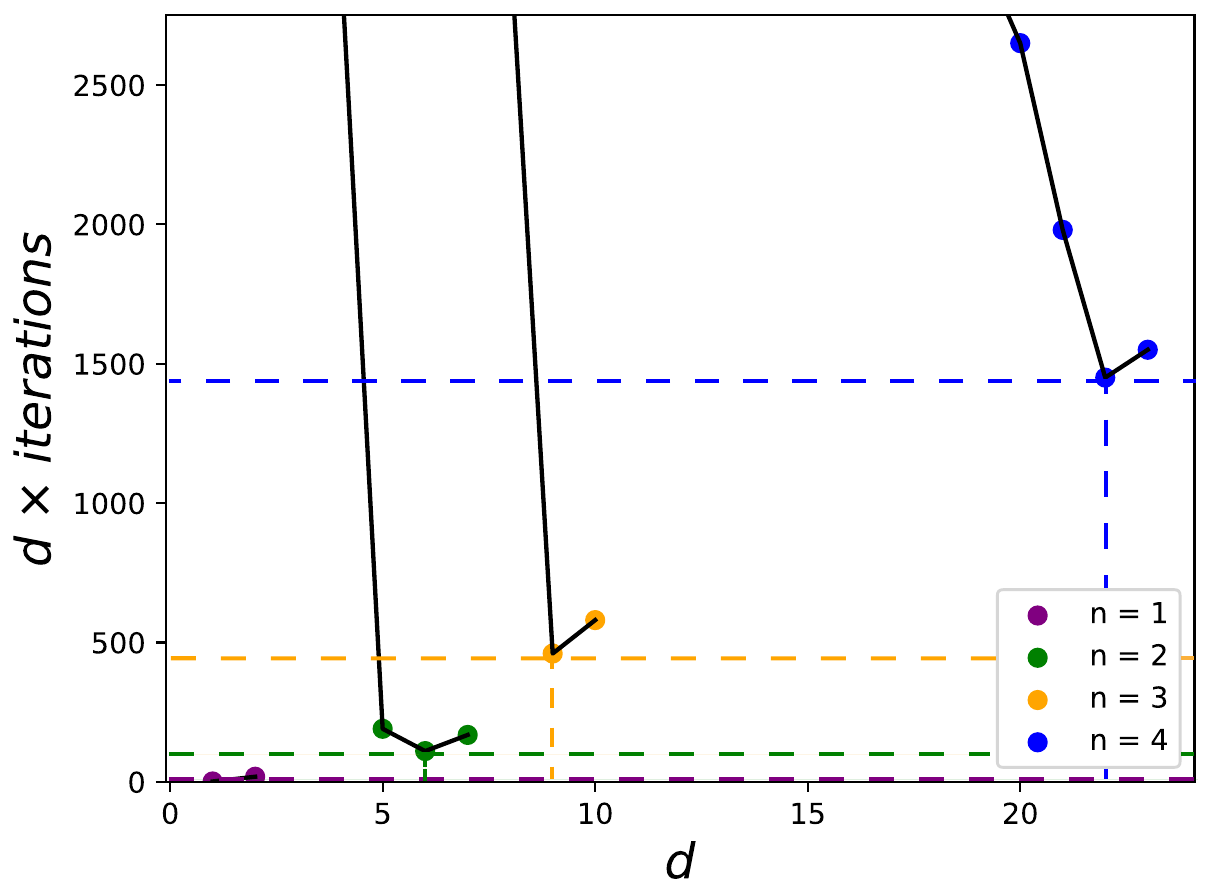}
\caption{Amount of required resources, as a function of the depth, characterized by the depth times the maximum number of iterations needed to learn 10 different sets of
eigenvalues and eigenvectors of a unitary
sampled from the Haar measure, with a cost
function taking, at most, the value 0.10.}

\label{Fig:Opt_depth}       
\end{figure} 

\section{Methods}
\label{Methods}
\subsection{Notation}
We denote vectors as, e.g., $\hat{\theta}$. $\hat{\theta_{j}}$ represents the $j_{th}$ block of $\hat{\theta}$, while $\theta_{j}$ denotes the $j_{th}$ parameter of $\hat{\theta}$. Quantum gates are defined with upper case letters, e.g., $RY(\theta)$. $U$ is the targeted  unitary and $U\textsubscript{VQC}$ is the unitary representing the variational quantum circuit.
A chain of quantum gates is denoted via a hyphen connecting two gates, e.g., $RY(\theta_0)$-$RX(\theta_1)$. For pure quantum states, we use the Dirac bra-ket notation $\ket{\psi}$ and lowercase Greek letters, e.g., \ $\rho$ and $\sigma$, for general density-matrix states. The variables \(t\) and \(a\) represent the number of target and ancilla qubits of the QPUF, respectively. We denote the chosen depth of the VQC as \(d\).

\subsection{Optimal depth}
\label{Op_methods}
The works of \cite{PhysRevA.101.052316} and \cite{Hai2023} relate the complexity of the VQC to the amount of qubits $n$ and depth $d$ via $\mathcal{O}(n^3d^2)$ and by creating a lower bound on the number of resources $R(n)\geq poly(n)$, respectively. However, such claims only provide a lower bound on the circuit complexity without ensuring a high performance. To offer a more practical understanding of the VQC complexity, we no longer treat \(d\) as being independent from \(n\), but rather expect to find some dependence \(d(n)\) for an efficient VQC performance. Recognizing this circuit behavior is essential for accurately assessing its implementability and scalability in near-term NISQ devices \cite{Bharti_2022}. To formally introduce the optimal depth, we begin by defining the concept of variational unitary space:

\begin{definition}
Variational Unitary Space: Space formed by all possible variational unitaries $U_{\mathrm{VQC}}$ attainable by a fixed architecture of the VQC $W(\hat{\theta_{j}})$ with $p$ real parameters

\begin{equation}
 \Gamma = \bigcup_{\hat{\theta}\in 
  {{\rm I\!R}}^{\times p}} \bigl\{ U_{\mathrm{VQC}}(\hat{\theta})\bigl\} \subseteq \mathbf{U(n)},
\end{equation}
\end{definition}

 where $\mathbf{U(n)}$ represents the unitary group of $n \times n$ matrices. 
  
  $U\textsubscript{VQC}$ can be described by block repetitions of length $d(n)$ which defines the depth of the VQC,
\begin{equation}
U\textsubscript{VQC}=\prod_{j=1}^{d(n)}W(\hat{\theta_{j}}),
\end{equation}
 where $W({\hat{\theta_{j}}})$ represents a combination of parametrized single-qubit and two-qubit gates. The definition of optimal depth is subsequently formulated:

\begin{definition}
Optimal depth ($D_{\mathrm{opt}}$): Assuming $\exists$  $d$ s.t. $U\subset\Gamma$, for a fixed architecture of the VQC $W({\hat{\theta_{j}}})$, $D_{opt}$ is the minimal amount of layers that suffices for a fiducial reconstruction of the unitary $U$. More concretely, with a fixed $W({\hat{\theta_{j}}})$,

\begin{equation}
D_{\mathrm{opt}}(n)=\mathrm{min}\{d(n)\} \hspace{0.2cm} \mathrm{s.t.}  \hspace{0.2cm} U_{\mathrm{VQC}}(\hat{\theta})\approx U.
\end{equation}
\end{definition}


\begin{figure}
\hspace{0.1cm}
\includegraphics[width=0.45\textwidth,clip]{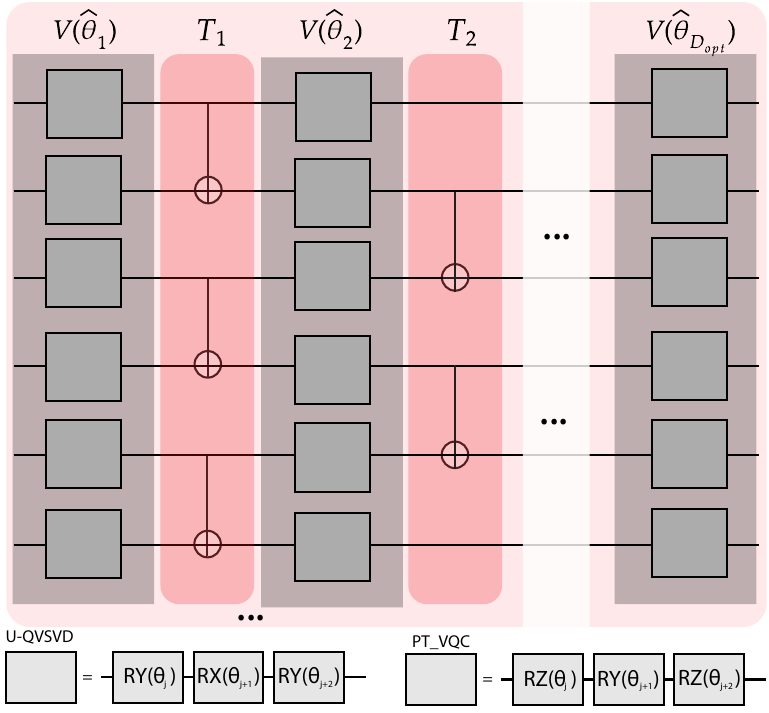}
\caption{Architecture of the VQC corresponding to $U\textsubscript{VQC}$. Different single qubit gate sequences were used for PT\_VQC and U-VQSVD algorithms.}
\label{ARR}       
\end{figure} 

For a given fixed value of 
$n$, each depth 
$d$ in the VQC corresponds to the number of iterations required to achieve the target accuracy of the cost function in the PT\_VQC algorithm. We search for the minimal mean number of iterations needed to attain this accuracy across 10 different target Haar random unitaries. The depth $d$ associated to it is considered as the optimal depth $D_{opt}$. In the U-VQSVD scheme, instead of considering the mean number of iterations, the maximum number of iterations times $d$ is recorded. This behavior is depicted in Fig. \ref{Fig:Opt_depth} within the context of U-VQSVD (Section \ref{Unitary singular value decomposition}) for the task of reconstructing the eigenvalues and eigenvectors of 10 Haar-random unitaries.


\subsection{PT\_VQC algorithm}
\label{Process tomography2}

A new method is proposed for VQC process tomography (PT\_VQC). We consider for our VQC the following parameterized architecture composed of single gate blocks $V(\hat{\theta_{j}})$ and two-qubit gate blocks $T_{j}$,

\begin{equation}
W(\hat{\theta_{j}})= V(\hat{\theta_{j}})T_{j}.
\end{equation}

\(U\textsubscript{VQC}(\hat{\theta})\) is constructed using \(d(n)\) blocks of single-qubit gate chains, comprising \(RZ\)-\(RY\)-\(RZ\) gates, and two-qubit gate blocks \(T_{j}=CX_{\mathrm{odd}_{j}|\mathrm{even}_{j}}\). Here, \(CX_{\mathrm{odd}_{j}|\mathrm{even}_j}\) represents a sequence of \(CX\) gates controlling either odd or even qubits, depending on the block number \(j\).
 
 The architecture for the VQC representing the $U\textsubscript{VQC}$ for the PT\_VQC and U-VQSVD algorithms can be seen in Fig.\,\ref{ARR} and the design of the algorithm is presented in Fig.\,\ref{architecture}. For a set of initializations completing an orthonormal basis $\{\ket{i}\}_{i=0}^{2^{n}-1}$, the state $\ket{\xi}$ before measurement has the form

\begin{equation}
\ket{\xi}=\frac{1}{2}\left(\ket{\Psi_{i}^{\mathrm{pr}}}+\ket{\Psi_{i}^{\mathrm{tr}}}\right) \otimes \ket{0}+\frac{1}{2}\left(\ket{\Psi_{i}^{\mathrm{pr}}}-\ket{\Psi_{i}^{\mathrm{tr}}}\right) \otimes \ket{1},
\label{calcevol2}
\end{equation}

where $\ket{\Psi_{i}^{\mathrm{pr}}} = U\textsubscript{VQC}(\hat{\theta})\ket{i}$
and $\ket{\Psi_{i}^{\mathrm{tr}}} = U\ket{i}$.

To define the subsequent partial measurement, the Von Neumann formalism is used, with the measurement projector  $\Omega_{0}=I\otimes\ket{0}\bra{0}$ being related to the probability $p_0$ of obtaining the outcome "0" in the following manner

\begin{align*}
p_{0}=\bra{\xi}{(\Omega_{0})^{\dagger}(\Omega_{0})}\ket{\xi}\\
=\frac{1}{2} - \frac{1}{2} \Re\left(\bra{\Psi_{i}^{\mathrm{tr}}}\ket{\Psi_{i}^{\mathrm{pr}}}\right). \tag{6}
\label{calcevol}
\end{align*}

\begin{figure}
    \begin{quantikz}
    &&\lstick{$\ket{\Psi}$} &\qw&\gate{\mathrm{U}}&\qw  &  \gate{\mathrm{U_{VQC}}} & \qw &  \qw  \\
    &&\lstick{$\ket{0}$}&\gate{H}&\ctrl{-1} &\qw & \octrl{-1} &\gate{H} & \meter{} & 
    \end{quantikz}
    \caption{Architecture for the proposed VQC tomography algorithm (PT\_VQC).}
\label{architecture}
\end{figure}
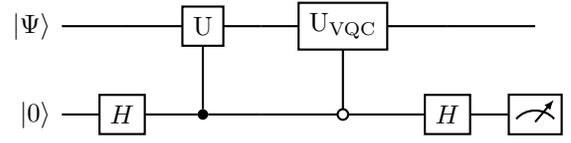



We define the cost function as 
\begin{align*}
C(\hat{\theta}) &= \mathrm{\frac{1}{2^{n}}\sum_{i=1}^{2^{n}} \norm{\ket{\Psi_{i}^{\mathrm{pr}}}-\ket{\Psi_{i}^{\mathrm{tr}}}}^2} \tag{7}\\
&=
\frac{1}{2^n} \sum_{i=1}^{2^n} \Bigg(\braket{\Psi_{i}^{\mathrm{pr}}}{ \Psi_{i}^{\mathrm{pr}}}+
\braket{\Psi_{i}^{\mathrm{tr}}}{\Psi_{i}^{\mathrm{tr}}}\\
& + \braket{\Psi_{i}^{\mathrm{pr}}}{\Psi_{i}^{\mathrm{tr}}}   \braket{\Psi_{i}^{\mathrm{tr}}}{ \Psi_{i}^{\mathrm{pr}}}\Bigg)\tag{8}  \\ \tag{9}
&= \mathrm{\frac{1}{2^{n-1}} \sum_{i=1}^{2^{n}} \left[1-\Re \left(\braket{\Psi_{i}^{\mathrm{tr}}}{\Psi_{i}^{\mathrm{pr}}}\right)\right]}. 
\label{3}
\end{align*}

In Appendix.\,\ref{Cost_proof}, we provide the mathematical justification for the choice of this cost function. Its computation comes from the measurement over a controlled version of the VQC, thus its gradient is computed by the 4-term shift rule \cite{Wierichs_2022}

    \begin{align*} \nabla_{\theta_{i}}C(\hat{\theta}) =  
      \frac{\sqrt{2}+1}{4\sqrt{2}}[C(\hat{\theta})\rvert_{\theta_i = \theta_i + \frac{\pi}{2}}-C(\hat{\theta})\rvert_{\theta_i =\theta_i-\frac{\pi}{2}}]-\\
      \frac{\sqrt{2}-1}{4\sqrt{2}}[C(\hat{\theta})\rvert_{\theta_i= \theta_i + \frac{3\pi}{2}}-C(\hat{\theta})\rvert_{\theta_i=\theta_i - \frac{3\pi}{2}}].
      \tag{10}
    \label{parameter}
\end{align*}

 The target unitary $U$ is sampled from the Haar measure using the QR decomposition. In the minimization process, we employ the Adam optimizer, with hyperparameters \( \beta_1 = 0.8 \) and \( \beta_2 = 0.999 \), which respectively denote the decay rates of the first and second-moment estimates. In Appendix.\,\ref{appcode}, we present the pseudo-code summarizing the algorithm routines for PT\_VQC.

\subsection{PT\_VQC comparison to QDNN and Tensor Networks}
In this section, we outline the methodology utilized for comparing PT\_VQC, QDNN \cite{Beer2020} and tensor networks \cite{Torlai2023} using MPOs. We set the number of iterations to $60$ and establish the target average fidelity as

\begin{equation*}
    F=\frac{1}{2^{n}}\sum_{i=1}^{2^{n}}({tr\sqrt{\sqrt{\rho_{i}}\sqrt{\rho^{'}_{i}}\sqrt{\rho_{i}}}})^{2}\geq0.9, \tag{11}
\end{equation*}
where $\rho$ and $\rho^{'}$ represent the density matrices of the ideal and predicted output states, respectively. It is important to note that we solely rely on fidelity as a metric for comparing PT\_VQC to tensor networks and QDNN. The former cost function definitions remain consistent for the PT\_VQC and U-VQSVD algorithms. Afterwards, we determine the minimum input state requirements for each algorithm ensuring the proper minimization of the cost function. To ensure the robustness of our results, we assess these requirements across 5 different seeds and calculate their standard deviations.

For variational tensor networks \cite{Torlai2023}, the cost function relies on minimizing the Kullback-Leibler divergence \cite{Clim2018}. The strategy behind this algorithm is to construct the MPO representation of \(U\), assuming structures with low entanglement degrees. During each training step, we evaluate the fidelities between the predicted and actual state vectors. Employing a batch size of 10 samples, we utilize Symmetric Informationally Complete Positive Operator-Valued Measures \cite{Rastegin_2014} (SIC POVMs), producing $4^{n}$ state initializations.

For QDNN \cite{Beer2020}, the cost function is defined by the fidelity between the QDNN output  $\rho_{\mathrm{x}}^{\mathrm{out}}$ and the correct output $\phi_{\mathrm{x}}^{\mathrm{out}}$ averaged over the training data size $N$

\begin{equation*}
    C_{QDNN}=\frac{1}{N}\sum_{\mathrm{x}=1}^N \expval{\rho_{\mathrm{x}}^{\mathrm{out}}}{\phi_{\mathrm{x}}^{\mathrm{out}}} \tag{12}.
\end{equation*}
To meet the criterion of achieving a fidelity of 0.9 across 5 consecutive seeds, we utilize $3^n$ initial states for the QDNN scheme. 

\subsection{U-VQSVD algorithm}
\label{Unitary singular value decomposition}

\begin{figure}
    \begin{quantikz}
    &&&&
    \lstick{$\ket{0}$} &\gate{F^\dagger}&\ctrl{2}  &  \gate{F} & \meter{k}   
    \\
    &&&&
    &           &             &           &             &          
    \\
    &&&&
    \lstick{$\ket{\Psi}$}&\qw &\gate{U}& \qw & \qw \ket{\Psi_{k}} &    
    \end{quantikz}
    \caption{PE-QPUF architecture:  For a target $t$ qubit initialization $\ket{\Psi}$ and an ancilla $a$ qubit initialization $\ket{0}$, a set of classical outputs \(k\) and partially measured states \(\ket{\Psi_{k}}\) are generated. 
 During user verification, the partially measured state \(\ket{\Psi_{k}}\) is initialized, and the new classical output \(k'\) is compared to the one generated during initialization \(k\) by the trusted party. \(U\) defines the target unitary, and \(F\) represents the quantum Fourier transform, which, abusing notation, maps  $\ket{j}\mapsto	\frac{1}{2^{n/2}}\sum_{k=0}^{2^{n}-1}e^{\frac{2\pi i j k}{2^{n}}}\ket{k}$.}
\label{fig:QEQPUF}
\end{figure}
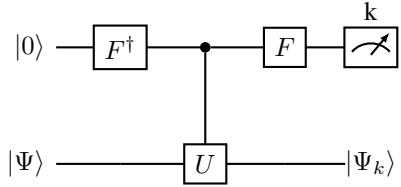

In this section, we introduce the U-VQSVD algorithm, drawing inspiration from the method for singular value decomposition of a unitary $U$ delivered by  \cite{Wang_2021}. Similar to the unitary process tomography discussed in Section \ref{Process tomography2}, the architecture of \(U_{\text{VQC}}(\hat{\theta})\) follows that of PT\_VQC, with the only change being the sequence of single-qubit gate chains, now arranged as \(RY\)-\(RX\)-\(RY\) as depicted in Fig.\,\ref{ARR}.

Using the U-VQSVD scheme, we propose a new method for process tomography for unknown general channels assuming that the corresponding eigenstates can be learned up to a phase. As a demonstration of the U-VQSVD performance we conduct a forgery attempt on a PE-QPUF \cite{QPUF_part}, with its architecture depicted in Fig. \ref{fig:QEQPUF}. The PE-QPUF is constructed with a Haar-random unitary \( U \) which, in the user verification stage, probabilistically yields a classical output \( k' \) close to the one generated during generation, \( k \) \cite{QPUF_part}. Both stages, generation and verification, are constructed using the quantum phase estimation algorithm. Consequently, the architecture comprises two subsets of qubits: ancillary and target qubits. Fig. \ref{fig:QEQPUF} illustrates the circuit responsible for generating the initial classical output $k$ with respective quantum state $\ket{\psi_{k}}$ to be acquired by the user. Conversely, during verification, the post-measurement state obtained during the generation serves as the input for the target system. Typically, the initialization state is not an eigenstate of \( U \) due to its Haar-random nature.

\subsection{U-VQSVD attack on PE-QPUF}
We outline the steps employed by U-VQSVD to numerically execute an impersonation attack on the PE-QPUF authentication protocol:

\begin{figure}
\includegraphics[width=0.45\textwidth,clip]{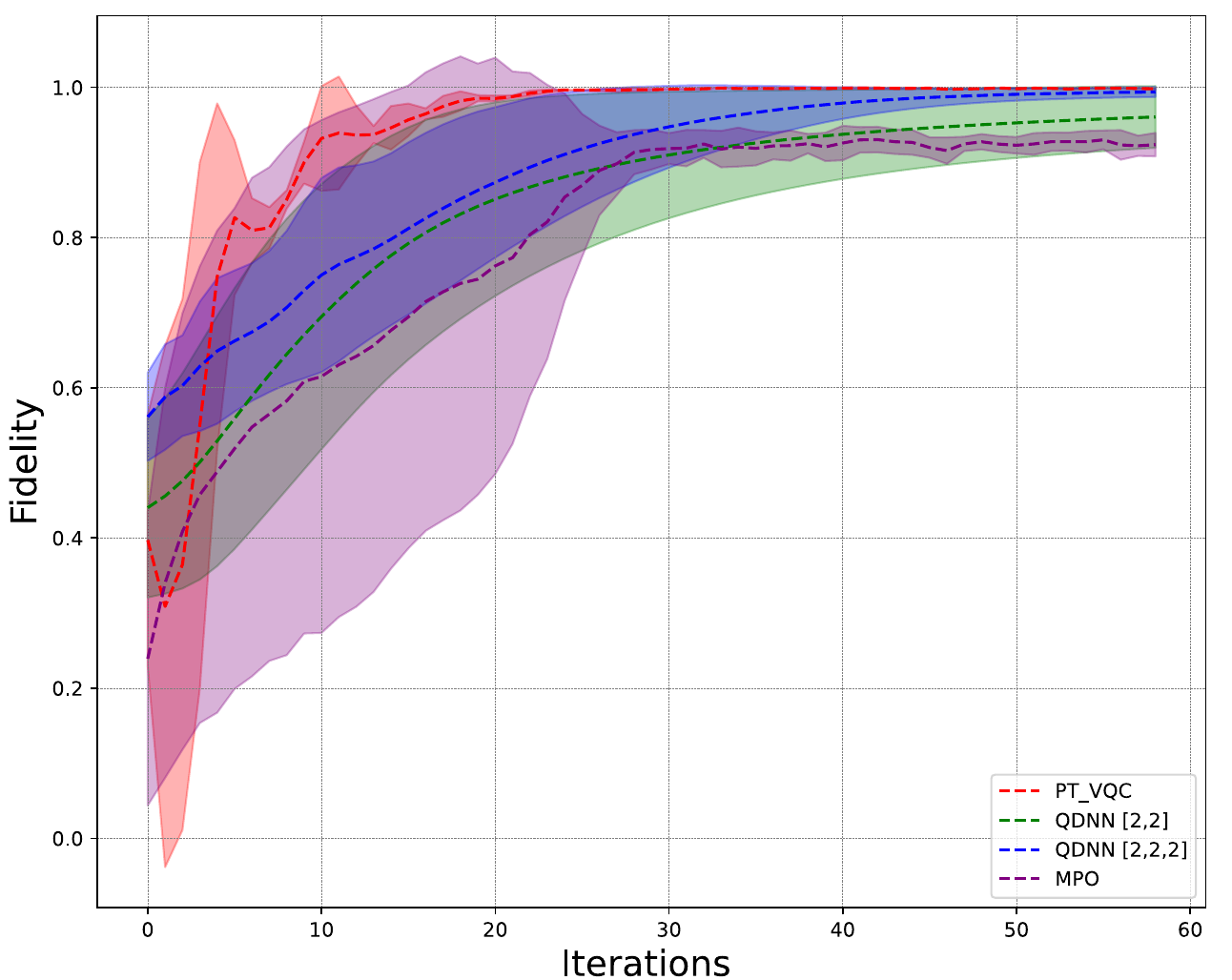}
\caption{Comparison between PT\_VQC for 2 qubits with variational MPO and QDNN with (blue) and without (green) a hidden layer in the structure. We compute the standard deviation based on 5 different Haar random unitaries.}
\label{QDNN_proc}       
\end{figure}

\begin{figure*}

\begin{minipage}[t]{0.495\textwidth}
\includegraphics[width=1.0\linewidth]
{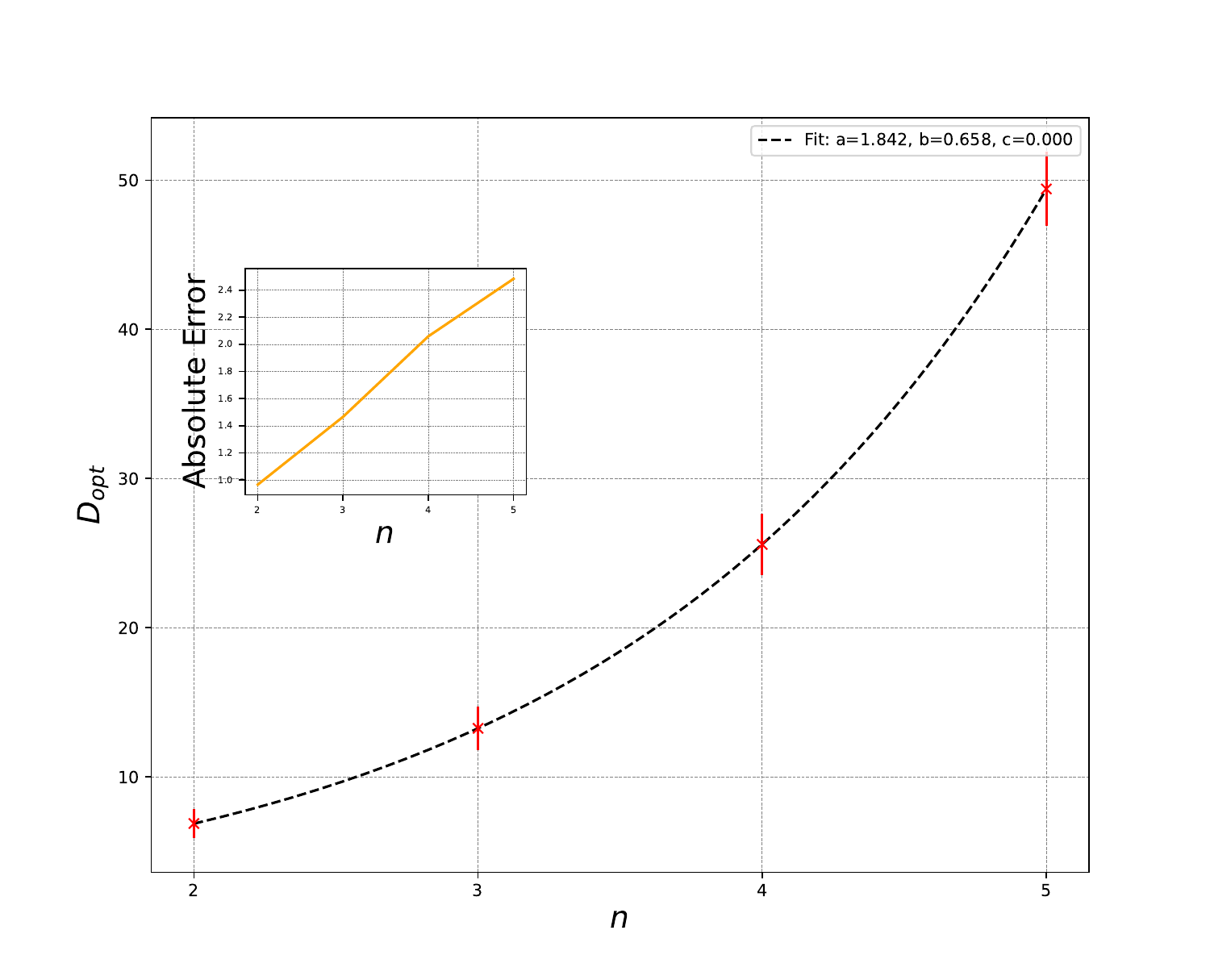}
  \end{minipage}
  \hfill
  \begin{minipage}[t]{0.495\textwidth}
    \includegraphics[width=1.0\linewidth]{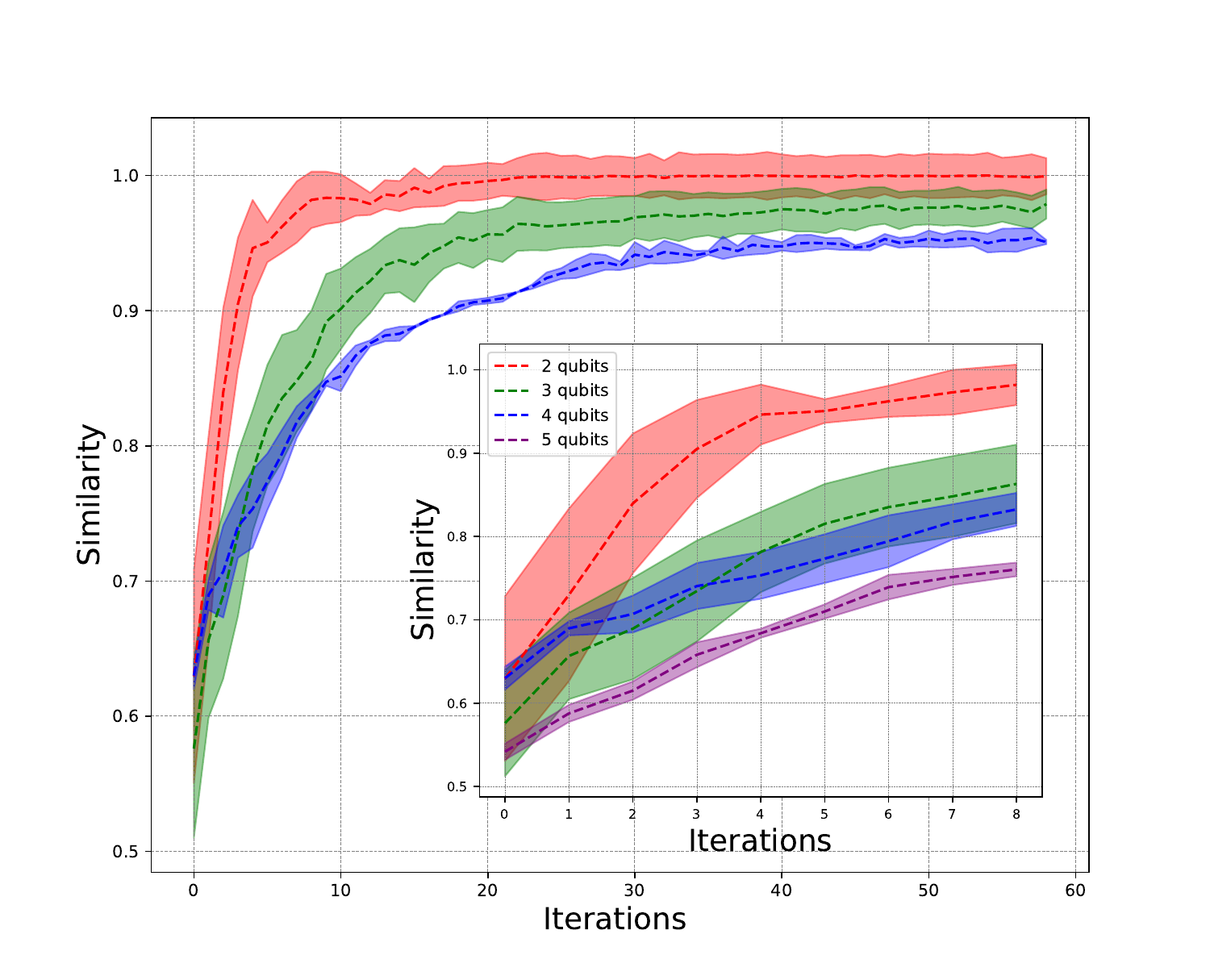}
  \end{minipage}
  
  \caption{Process tomography performance (PT\_VQC): Left) Optimal depth ($D_{opt}$) fit for the minimal resources needed to achieve maximum fidelity with the respective absolute error. Right) Similarity for 5 different Haar-random unitaries per qubit in terms of the training iteration number.}
\label{Similarity}
\end{figure*}

\begin{enumerate}
  \item We start by preparing $2^n$ states completing the computational basis $\{\ket{i}\}_{i=0}^{2^{n}-1}$ and evolve each of them via the VQC $U\textsubscript{VQC}$, using the same methodology as for the PT\_VQC algorithm in Section \ref{Process tomography2}.
  \item The cost function, instead of being defined with two independent VQCs, as in the work of \cite{Wang_2021}, is just defined with one
  
  \begin{equation*}
    C(\hat{\theta})=\frac{1}{2^{n}}\sum _{i=0}^{2^n-1} 1-  \abs{\expval{U\textsubscript{VQC}(\hat{\theta})^{\dagger}UU\textsubscript{VQC}(\hat{\theta})}{i}}. \tag{13}
    \label{cost_svd}
    \end{equation*}
    Using the modified definition in Eq.\,\ref{cost_svd}, the cost of the gradient computation complexity is reduced. This is achieved as the objective has changed, from ideally reconstructing the singular vectors of $U$, to learning them up to a global phase. 
  
    \item The parameters \(\hat{\theta}\) of the VQC are updated using the Adam optimizer. Consequently, we obtain $\{ \ket{\phi_i} \}_{i=0}^{2^n - 1}$ as the estimated singular vectors of $U$, where $\ket{\phi_i} = U_{\text{VQC}}(\hat{\theta})\ket{i}$ and $e^{i\phi_i} ={\expval{U_{\text{VQC}}(\hat{\theta})^{\dagger}UU_{\text{VQC}}(\hat{\theta})}{i}}$ is  the $i$-th estimated  singular value.

  \item Once the singular values and vectors of \(U\) have been learned, the attack occurs during the verification stage. The objective of the attacker is to impersonate the user by sending the learned singular state \(\ket{\phi_i}\), corresponding to the complex phase $\phi_i$  which minimizes the quantity \(\Big|\frac{\phi_i}{2 \pi} - \frac{k}{2^a} \Big|\).
\end{enumerate}

We assess the efficacy of this novel attack by comparing its capabilities to those of an attacker utilizing random states during the verification procedure, as well as to those of a trusted party undergoing verification.

\section{Results}
\label{Results}



\subsection{PT-VQC algorithm}
\label{Process tomography}

In Fig. \ref{QDNN_proc}, we compare the performance of the PT\_VQC algorithm with QDNN, both with (blue) and without (green) hidden layers \cite{Beer2020}, as well as tensor networks schemes \cite{Torlai2023}. For a scenario involving 2 qubits, with a fixed number of 60 iterations and a target fidelity of 0.9, we observe a faster convergence behavior in our scheme compared to the other two methods. It is worth noting that the attenuated oscillatory behavior observed for PT\_VQC in Fig. \ref{QDNN_proc} arises from the minimization problem being formulated with respect to the cost function defined in Eq. \,\ref{3}, which is not a monotonic function of the average fidelity.

The optimal depth of the VQC for a varying amount of qubits is presented in Fig.\,\ref{Similarity} (left).
To model the relationship between the optimal depth and the number of qubits, we conducted an exponential fit of the form \(D_{\text{opt}}(n) = ae^{bn} + c\), considering the standard deviation of the optimal depth, which was calculated for 5 different Haar random unitaries $U$. In Fig. \ref{Similarity} (right), we observe the similarity, as defined in Appendix \ref{Cost_proof}, from 2 to 5 qubits over varying iteration numbers for 5 different Haar-random unitaries \(U\). The performance in Fig. \ref{Similarity} (right) exhibits a convergent behavior, where the amplitude decreases with an increase in the number of qubits. For 5 qubits, only the first 9 iterations were generated due to high processing times. Nevertheless, Fig. \ref{Similarity} (left) enables us to extrapolate the optimal depth of the VQC for larger qubit sizes.



\subsection{U-VQSVD algorithm}
\label{Unitary singular value decomposition2}

Following the attack presented in Section \ref{Unitary singular value decomposition}, which relies on properly learning the singular values and vectors of the target unitary \(U\) of the PE-QPUF, we aim to forge the identity of \(10^{2}\) different users. We sample \(10^{2}\) Haar-random unitaries for system sizes ranging from 1 to 4 qubits. For each \(U\), we choose 5 different ancilla sizes \(a\) ranging from 2 to 6 qubits, resulting in a total of \(4 \times 5 = 20\) PE-QPUFs.

In each scenario, we conduct \(25 \times 2 ^a\) forgeries and compute the average discrepancy between the generation and verification outcomes. The data points depicted in Fig. \ref{SVD} represent the mean of this average across the entire set of \(10^2\) different users, with the error bars indicating the fluctuations observed in the comparison across multiple users. In Fig. \ref{SVD} (left), we observe the mean deviation of the U-VQSVD attack, contrasting it with Fig. \ref{SVD} (right), where an attack with random initializations $\ket{\Psi_{k}}$ is performed during the verification. Additionally, in Fig. \ref{SVD} (middle), we compare it to the reference mean deviation, where a trusted party performs the authentication.

\begin{figure*}
\includegraphics[width=1.0\textwidth,clip]{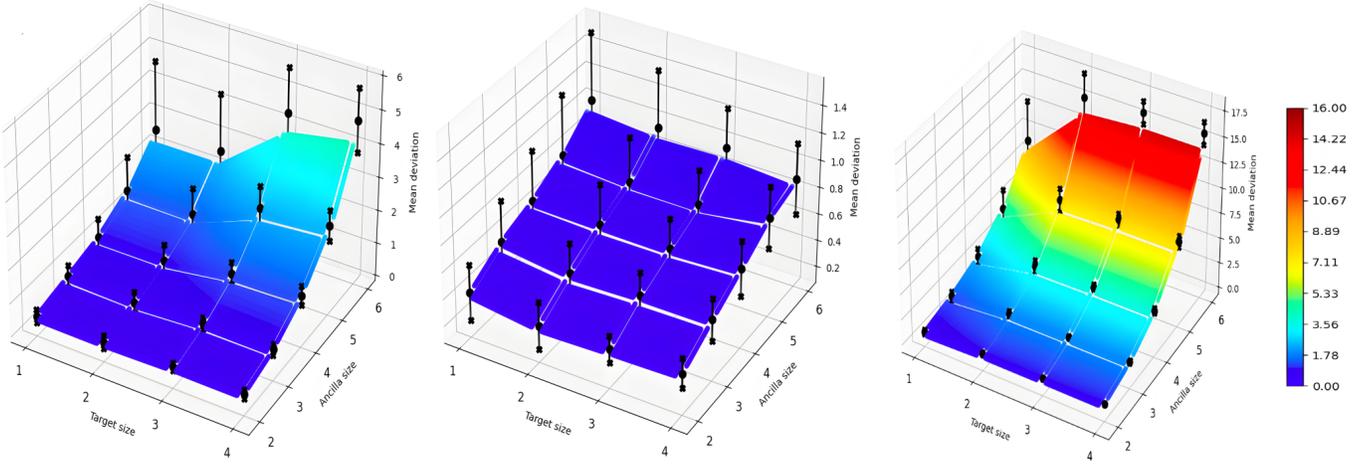}
\caption{SVD attack performance. The plots show the mean deviation between generation and verification outcomes for 20 different PE-QPUFs instances. Left) Displays the results coming out from the SVD attack. Middle) was obtained by simulating user-server interactions. Right) one can benchmark our attack with the performance obtained by means of a random forger, that is, always sending the state $\ket{0}^{\otimes t}$, with $t$ being the number of target target qubits. }
\label{SVD}       
\end{figure*} 

In Fig. \ref{SVD} we validate the effectiveness of the U-VQSVD targeting unknown quantum unitary evolutions. In other words, it demonstrates how an attack using our U-VQSVD algorithm, directed towards the previously defined PE-QPUF, outperforms an uninformed attack by a factor of 2 to 5 depending on the $t$ and $a$ sizes. It is important to note that this does not undermine the security notions established in \cite{QPUF_part}. In fact, the trusted party-server interaction demonstrates better behavior compared to our forgery attempts. For larger sizes \(t\) and \(a\), for which security is proven to hold, our attack would already fail under its own assumptions, implying that the resources required for such a scheme scale exponentially.


\label{Training data generation}


\section{Conclusions}
\label{Discussion}

In this study, we introduce two innovative process tomography techniques: PT\_VQC and U-VQSVD.

PT\_VQC (Section.\,\ref{Process tomography2}) compared to \cite{PhysRevA.105.032427} significantly reduces the required number of state initializations to perform process tomography from \(4^{n}\) to \(2^{n}\) and halves the necessary number of qubits for each run. This enhancement results in improved algorithm processing time and increased unitary reconstruction accuracy. PT\_VQC surpasses both QDNN and tensor networks (MPO), as demonstrated in Fig.\,\ref{QDNN_proc}. Unlike tensor networks utilizing MPOs, which assume low entanglement structures, PT\_VQC performs well regardless of the structure of \(U\).

U-VQSVD (Section.\,\ref{Unitary singular value decomposition}) presents a novel process tomography scheme designed for unknown general channels. It assumes that the corresponding eigenstates can be learned up to a global phase and is based on the variational quantum singular value decomposition (VQSVD) architecture proposed in \cite{Wang_2021}. To evaluate the efficacy of the U-VQSVD algorithm, we conduct an impersonation attack on a PE-QPUF \cite{Mina21}. We compare the capabilities of an uninformed attacker, employing random states, with those of an attacker using U-VQSVD to generate input states, alongside a trusted party. U-VQSVD outperforms the uninformed attack by a factor of 2 to 5, depending on the ancilla and target sizes of the PE-QPUF.

It is worth noting that all the VQC employed have successfully provided solutions to the problem following a gradient-based minimization of the cost function. We conjecture that the optimization problems posed by our cost functions, defined in Eq.\,\ref{3}, Eq.\,\ref{cost_svd}, and the parameterized ansätze are not ill-posed. In other words, they do not possess undesirable local minima. The proof supporting this conjecture may pave the way for a new theoretical research direction in VQC optimization.



\appendix
\label{Appendix A}
\section*{Similarity correlation to the cost function}
\label{Cost_proof}

\begin{theorem} 
The complementary of the similarity  between the target and learned  operators  is  a monotonically increasing function of the cost defined in  Eq.\,\ref{3}. 
The vanishing of the cost function  defined in  Eq.\,\ref{3} implies a correct learning of the targeted operator.
\end{theorem}

\begin{proof}

The complementary of the similarity, $S(A,B)$, between two operators $A$ and $B$ is defined  via the Frobenius norm \cite{BOTTCHER20081864} as

\begin{equation*}
 1-S(A,B)=\frac{\norm{A-B}_{F}}{\norm{A}_{\mathrm{F}}}, \tag{14}
\end{equation*}

where, if $x_{ij}$ are the entries of an operator $X$ in an orthonormal basis $\{ \ket{i} \}_{i=0}^{2^n-1}$, then

\begin{equation*}
 \mathrm{\norm{X}\mathrm{_F}= \sqrt{\sum_{i}^{2^{n}-1}\sum_{j}^{2^{n}-1} \abs{x_{ij}}^2}}.\tag{15}
\end{equation*}

For the unitary  operators $U$ and $U(\hat{\theta})_{\mathrm{VQC}}$, we have

\begin{equation*}
  1-S\left(U,U(\hat{\theta})\textsubscript{VQC}\right) \propto^+ \norm{ (U-U(\hat{\theta})\textsubscript{VQC})}_{\mathrm{F}}, \tag{16}
\end{equation*}

where $\propto^+$ stands for "proportional via a positive constant".

If we define $\bra{k}\left(U-U(\hat{\theta})\textsubscript{VQC}\right)\ket{i} \equiv \alpha_{k}^i$, then

\begin{equation*}
 1- S\left(U,U(\hat{\theta})\textsubscript{VQC}\right) \propto^+ \sqrt{\sum_{i=0}^{2^n-1} \sum_{k=0}^{2^n-1} |\alpha^i_k|^2} \tag{17}        ,
\end{equation*}

\begin{equation*}
     = \sqrt{\sum_{i=0}^{2^n-1} \bigg| \bigg|\Big(U-U(\hat{\theta})\textsubscript{VQC}\big) \Big) \ket{i} \bigg| \bigg|^2},     \tag{18}
\end{equation*}

leading to

\begin{equation*}
   1- S\left(U,U(\hat{\theta})\textsubscript{VQC}\right) \propto^+ \sqrt{2 \sum_{i=0}^{2^n-1}   1 - \Re\{ \bra{i}U^\dagger U(\hat{\theta})_{\mathrm{VQC}}\ket{i}\}}, \tag{19}
\end{equation*}

to finally obtain

\begin{equation*}
    1- S\left(U,U(\hat{\theta})\textsubscript{VQC}\right) \propto^+ \sqrt{f(\hat{\theta})}. \tag{20}
\end{equation*}

The fact that $g(x)=\sqrt{x}$ is a monotonically increasing function added to the fact that $\sqrt{0} = 0$, concludes the proof.
\end{proof}

\section*{Computation of the real and imaginary part of an overlap between quantum states}
\label{Appendix_singular}
Considering a Hilbert space of dimension $2^n$ with $\ket{\psi}$, $\ket{\phi}$ $\in$ $\mathcal{H}^{\otimes n}$. The inner product $\bra{\phi}\ket{\psi}$ gathers the real and imaginary parts that we want to obtain, with the assumption of owning several copies of both $\ket{\psi}$ and $\ket{\phi}$. Let $V$ and $W$ be unitary operators such that $V\ket{0}^{\otimes n} = \ket{\psi}$ and  $W\ket{0}^{\otimes n} = \ket{\phi}$.

\paragraph*{Real part} 
We can calculate the real part by carrying out the following circuit on ($n+1$)-qubits register

\begin{equation*}
    \Big(H\otimes 1\!\!1\Big)\Big(C_1W    \Big)\Big(C_0V  \Big) \Big(H \otimes 1\!\!1\Big)\ket{0}\textsubscript{ancilla} \otimes\ket{0}^{\otimes n}, \tag{21}
    \label{real1}
\end{equation*}

hence, obtaining

\begin{equation*}
\Re\{ \bra{\phi}\ket{\psi}\} = 2 p_0-1, \tag{22}
\label{real2}
\end{equation*}

where the subscript in $C_x$ stands for the choice of the state to be controlled and $p_0$ is the probability of obtaining the classical outcome "$0$" when measuring the ancillary qubit.

\paragraph*{Imaginary part} 
We can also calculate the imaginary part by carrying out the following circuit on ($n+1$)-qubits register

\begin{equation*}
    \Big(H \otimes 1\!\!1\Big)\Big(C_1W    \Big)\Big(C_0V  \Big) \Big(P\Bigg(\frac{3}{2}\pi\Bigg)H\otimes 1\!\!1\Big)\ket{0}_{ancilla} \otimes\ket{0}^{\otimes n},\tag{23}
    \label{imag1}
\end{equation*}

hence, obtaining

\begin{equation*}
\Im \{ \bra{\phi}\ket{\psi}\} = 2p_0 -1. \tag{24}
\label{imag2}
\end{equation*}

With the real and imaginary parts defined, the attacker can calculate $\ket{\phi}= W(\hat{\theta)}\ket{i}$ and $\ket{\psi} = UW(\hat{\theta)}\ket{i}$, in order to obtain the singular value $\lambda_{i}$ corresponding to the $i$-th singular vector.

\section*{Pseudo-code for PT\_VQC}
\label{appcode}
\
The routines for the PT\_VQC algorithm are described below. Our objective is to achieve a cost function $C(\hat{\theta}) \geq 0.10$ for 5 consecutive Haar-random unitaries (with a threshold of 5). The same pseudo-code applies for U-VQSVD, but instead of the 4-term shift rule, we use the 2-term shift rule \cite{Wierichs_2022}.
\begin{algorithm}[H]
\begin{algorithmic}[1]

\caption{PT\_VQC  algorithm}\label{PT_VQC algorithm}
\State \textbf{Input:} $2^{n}$ orthogonal states, $j=1$.

\For {$j\leq$ threshold}{
    
    \For {all iterations}{
    \State Generation of a $jth$ seed Haar-random controlled Unitary $U$ as a target.
    
    \For {all $\hat{\theta}$}{
        
        \For {all $2^{n}$}{ 
            
            \State Generate $4$ circuits $U\textsubscript{VQC}$ for $\theta_i \pm \frac{\pi}{2}$ and $\theta_i \pm \frac{3\pi}{2}$ as seen in Fig.\ref{architecture}.
            
            \For {all $4$ circuits}
            {\State Perform transpilation and $m$ shot  measurements.
            \EndFor}}

    \EndFor
    }
        
\State Obtain $C(\hat{\theta})\rvert_{\theta_i = \theta_i \pm \frac{\pi}{2}}$ and $C(\hat{\theta})\rvert_{\theta_i = \theta_i \pm \frac{3\pi}{2}}$ averaged over $2^{n}$ states.
    
\State Perform the 4-term shift rule described in Eq.\,\ref{parameter}.
        
\EndFor
}    
\State Obtain $\nabla_{\theta_{i}}C(\hat{\theta})$ vector for all $\theta_{i}$.

\State Use the Adam optimizer to update $\hat{\theta}$, for which $C(\hat{\theta})$ is minimized.

\EndFor

}

\If{$C(\hat{\theta}) \geq 0.10$}

\State Update $D\textsubscript{opt}$.
\State $j=0$.

\EndIf

\State $j=j+1$

\EndFor

\end{algorithmic}
\end{algorithm}

\section*{Code Availability}
The github repository can be found at: \href{https://github.com/terrordayvg/PT_VQC-Tomography}{https://t.ly/REETC}.

\section*{Acknowledgement}
Christian Deppe and Roberto Ferrara were supported by the German Federal Ministry of Education and Research (BMBF), Grant 16KIS1005.
Christian Deppe acknowledge the financial support by the Federal Ministry of Education and Research of Germany in the programme of ''Souver\"an. Digital. Vernetzt''. Joint project 6G-life, project identification number: 16KISK002.
Vladlen Galetsky, Soham Ghosh and Christian Deppe were supported by the BMBF Project 16KISQ038.
Christian Deppe and Roberto Ferrara are supported by the Munich Quantum Valley, which is supported by the Bavarian state government with funds from the Hightech Agenda Bayern Plus.

\bibliographystyle{IEEEtran}
\bibliography{main}

\begin{thebibliography}{10}
\providecommand{\url}[1]{#1}
\csname url@samestyle\endcsname
\providecommand{\newblock}{\relax}
\providecommand{\bibinfo}[2]{#2}
\providecommand{\BIBentrySTDinterwordspacing}{\spaceskip=0pt\relax}
\providecommand{\BIBentryALTinterwordstretchfactor}{4}
\providecommand{\BIBentryALTinterwordspacing}{\spaceskip=\fontdimen2\font plus
\BIBentryALTinterwordstretchfactor\fontdimen3\font minus \fontdimen4\font\relax}
\providecommand{\BIBforeignlanguage}[2]{{%
\expandafter\ifx\csname l@#1\endcsname\relax
\typeout{** WARNING: IEEEtran.bst: No hyphenation pattern has been}%
\typeout{** loaded for the language `#1'. Using the pattern for}%
\typeout{** the default language instead.}%
\else
\language=\csname l@#1\endcsname
\fi
#2}}
\providecommand{\BIBdecl}{\relax}
\BIBdecl

\bibitem{Chen2023}
\BIBentryALTinterwordspacing
S.~Chen, J.~Cotler, H.-Y. Huang, and J.~Li, ``The complexity of nisq,'' \emph{Nature Communications}, vol.~14, no.~1, p. 6001, Sep 2023. [Online]. Available: \url{https://doi.org/10.1038/s41467-023-41217-6}
\BIBentrySTDinterwordspacing

\bibitem{RevModPhys.94.015004}
\BIBentryALTinterwordspacing
K.~Bharti, A.~Cervera-Lierta, T.~H. Kyaw, T.~Haug, S.~Alperin-Lea, A.~Anand, M.~Degroote, H.~Heimonen, J.~S. Kottmann, T.~Menke, W.-K. Mok, S.~Sim, L.-C. Kwek, and A.~Aspuru-Guzik, ``Noisy intermediate-scale quantum algorithms,'' \emph{Rev. Mod. Phys.}, vol.~94, p. 015004, Feb 2022. [Online]. Available: \url{https://link.aps.org/doi/10.1103/RevModPhys.94.015004}
\BIBentrySTDinterwordspacing

\bibitem{Beer2020}
\BIBentryALTinterwordspacing
K.~Beer, D.~Bondarenko, T.~Farrelly, T.~J. Osborne, R.~Salzmann, D.~Scheiermann, and R.~Wolf, ``Training deep quantum neural networks,'' \emph{Nature Communications}, vol.~11, no.~1, p. 808, Feb 2020. [Online]. Available: \url{https://doi.org/10.1038/s41467-020-14454-2}
\BIBentrySTDinterwordspacing

\bibitem{Torlai2023}
\BIBentryALTinterwordspacing
G.~Torlai, C.~J. Wood, A.~Acharya, G.~Carleo, J.~Carrasquilla, and L.~Aolita, ``Quantum process tomography with unsupervised learning and tensor networks,'' \emph{Nature Communications}, vol.~14, no.~1, p. 2858, May 2023. [Online]. Available: \url{https://doi.org/10.1038/s41467-023-38332-9}
\BIBentrySTDinterwordspacing

\bibitem{PhysRevA.105.032427}
\BIBentryALTinterwordspacing
S.~Xue, Y.~Liu, Y.~Wang, P.~Zhu, C.~Guo, and J.~Wu, ``Variational quantum process tomography of unitaries,'' \emph{Phys. Rev. A}, vol. 105, p. 032427, Mar 2022. [Online]. Available: \url{https://link.aps.org/doi/10.1103/PhysRevA.105.032427}
\BIBentrySTDinterwordspacing

\bibitem{PhysRevResearch.5.L032040}
\BIBentryALTinterwordspacing
S.~Liu, S.-X. Zhang, S.-K. Jian, and H.~Yao, ``Training variational quantum algorithms with random gate activation,'' \emph{Phys. Rev. Res.}, vol.~5, p. L032040, Sep 2023. [Online]. Available: \url{https://link.aps.org/doi/10.1103/PhysRevResearch.5.L032040}
\BIBentrySTDinterwordspacing

\bibitem{Hai2023}
\BIBentryALTinterwordspacing
V.~T. Hai and L.~B. Ho, ``Universal compilation for quantum state tomography,'' \emph{Scientific Reports}, vol.~13, no.~1, p. 3750, Mar 2023. [Online]. Available: \url{https://doi.org/10.1038/s41598-023-30983-4}
\BIBentrySTDinterwordspacing

\bibitem{PhysRevResearch.3.L032049}
\BIBentryALTinterwordspacing
Y.~Wu, J.~Yao, P.~Zhang, and H.~Zhai, ``Expressivity of quantum neural networks,'' \emph{Phys. Rev. Res.}, vol.~3, p. L032049, Aug 2021. [Online]. Available: \url{https://link.aps.org/doi/10.1103/PhysRevResearch.3.L032049}
\BIBentrySTDinterwordspacing

\bibitem{Moll_2018}
\BIBentryALTinterwordspacing
N.~Moll, P.~Barkoutsos, L.~S. Bishop, J.~M. Chow, A.~Cross, D.~J. Egger, S.~Filipp, A.~Fuhrer, J.~M. Gambetta, M.~Ganzhorn, A.~Kandala, A.~Mezzacapo, P.~Müller, W.~Riess, G.~Salis, J.~Smolin, I.~Tavernelli, and K.~Temme, ``Quantum optimization using variational algorithms on near-term quantum devices,'' \emph{Quantum Science and Technology}, vol.~3, no.~3, p. 030503, jun 2018. [Online]. Available: \url{https://doi.org/10.1088%2F2058-9565%2Faab822}
\BIBentrySTDinterwordspacing

\bibitem{HUANG2020286}
\BIBentryALTinterwordspacing
X.-L. Huang, J.~Gao, Z.-Q. Jiao, Z.-Q. Yan, Z.-Y. Zhang, D.-Y. Chen, X.~Zhang, L.~Ji, and X.-M. Jin, ``Reconstruction of quantum channel via convex optimization,'' \emph{Science Bulletin}, vol.~65, no.~4, pp. 286--292, 2020. [Online]. Available: \url{https://www.sciencedirect.com/science/article/pii/S2095927319306413}
\BIBentrySTDinterwordspacing

\bibitem{PhysRevResearch.6.013029}
\BIBentryALTinterwordspacing
R.~Levy, D.~Luo, and B.~K. Clark, ``Classical shadows for quantum process tomography on near-term quantum computers,'' \emph{Phys. Rev. Res.}, vol.~6, p. 013029, Jan 2024. [Online]. Available: \url{https://link.aps.org/doi/10.1103/PhysRevResearch.6.013029}
\BIBentrySTDinterwordspacing

\bibitem{Mina21}
\BIBentryALTinterwordspacing
M.~Arapinis, M.~Delavar, M.~Doosti, and E.~Kashefi, ``Quantum {P}hysical {U}nclonable {F}unctions: {P}ossibilities and {I}mpossibilities,'' \emph{{Quantum}}, vol.~5, p. 475, Jun. 2021. [Online]. Available: \url{https://doi.org/10.22331/q-2021-06-15-475}
\BIBentrySTDinterwordspacing

\bibitem{cry}
\BIBentryALTinterwordspacing
B.~Skoric, ``Quantum readout of physical unclonable functions: Remote authentication without trusted readers and authenticated quantum key exchange without initial shared secrets,'' Cryptology ePrint Archive, Paper 2009/369, 2009, \url{https://eprint.iacr.org/2009/369}. [Online]. Available: \url{https://eprint.iacr.org/2009/369}
\BIBentrySTDinterwordspacing

\bibitem{Pirnay_2022}
\BIBentryALTinterwordspacing
N.~Pirnay, A.~Pappa, and J.-P. Seifert, ``Learning classical readout quantum pufs based on single-qubit gates,'' \emph{Quantum Machine Intelligence}, vol.~4, no.~2, Jun. 2022. [Online]. Available: \url{http://dx.doi.org/10.1007/s42484-022-00073-1}
\BIBentrySTDinterwordspacing

\bibitem{GGDF22}
V.~Galetsky, S.~Ghosh, C.~Deppe, and R.~Ferrara, ``Comparison of quantum puf models,'' in \emph{2022 IEEE Globecom Workshops (GC Wkshps)}, 2022, pp. 820--825.

\bibitem{QPUF_part}
S.~Ghosh, V.~Galetsky, P.~J. Farré, C.~Deppe, R.~Ferrara, and H.~Boche, ``Existential unforgeability in quantum authentication from quantum physical unclonable functions based on random von neumann measurement,'' 2024.

\bibitem{Wang_2021}
\BIBentryALTinterwordspacing
X.~Wang, Z.~Song, and Y.~Wang, ``Variational quantum singular value decomposition,'' \emph{Quantum}, vol.~5, p. 483, jun 2021. [Online]. Available: \url{https://doi.org/10.22331%2Fq-2021-06-29-483}
\BIBentrySTDinterwordspacing

\bibitem{PhysRevA.101.052316}
\BIBentryALTinterwordspacing
Y.~Liu, D.~Wang, S.~Xue, A.~Huang, X.~Fu, X.~Qiang, P.~Xu, H.-L. Huang, M.~Deng, C.~Guo, X.~Yang, and J.~Wu, ``Variational quantum circuits for quantum state tomography,'' \emph{Phys. Rev. A}, vol. 101, p. 052316, May 2020. [Online]. Available: \url{https://link.aps.org/doi/10.1103/PhysRevA.101.052316}
\BIBentrySTDinterwordspacing

\bibitem{Bharti_2022}
\BIBentryALTinterwordspacing
K.~Bharti, A.~Cervera-Lierta, T.~H. Kyaw, T.~Haug, S.~Alperin-Lea, A.~Anand, M.~Degroote, H.~Heimonen, J.~S. Kottmann, T.~Menke, W.-K. Mok, S.~Sim, L.-C. Kwek, and A.~Aspuru-Guzik, ``Noisy intermediate-scale quantum algorithms,'' \emph{Reviews of Modern Physics}, vol.~94, no.~1, Feb. 2022. [Online]. Available: \url{http://dx.doi.org/10.1103/RevModPhys.94.015004}
\BIBentrySTDinterwordspacing

\bibitem{Wierichs_2022}
\BIBentryALTinterwordspacing
D.~Wierichs, J.~Izaac, C.~Wang, and C.~Y.-Y. Lin, ``General parameter-shift rules for quantum gradients,'' \emph{Quantum}, vol.~6, p. 677, mar 2022. [Online]. Available: \url{https://doi.org/10.22331%2Fq-2022-03-30-677}
\BIBentrySTDinterwordspacing

\bibitem{Clim2018}
\BIBentryALTinterwordspacing
A.~Clim, R.~D. Zota, and G.~Tinic{\u{A}}, ``The kullback-leibler divergence used in machine learning algorithms for health care applications and hypertension prediction: A literature review,'' \emph{Procedia Computer Science}, vol. 141, pp. 448--453, Jan 2018. [Online]. Available: \url{https://www.sciencedirect.com/science/article/pii/S1877050918317939}
\BIBentrySTDinterwordspacing

\bibitem{Rastegin_2014}
\BIBentryALTinterwordspacing
A.~E. Rastegin, ``Notes on general sic-povms,'' \emph{Physica Scripta}, vol.~89, no.~8, p. 085101, jun 2014. [Online]. Available: \url{https://dx.doi.org/10.1088/0031-8949/89/8/085101}
\BIBentrySTDinterwordspacing

\bibitem{BOTTCHER20081864}
\BIBentryALTinterwordspacing
A.~Böttcher and D.~Wenzel, ``The frobenius norm and the commutator,'' \emph{Linear Algebra and its Applications}, vol. 429, no.~8, pp. 1864--1885, 2008. [Online]. Available: \url{https://www.sciencedirect.com/science/article/pii/S0024379508002772}
\BIBentrySTDinterwordspacing

\end{thebibliography}
\end{document}